\documentclass[submission,copyright,creativecommons]{eptcs}
\providecommand{\event}{SOS 2007} 
\usepackage{breakurl}             
\usepackage{underscore}           
\usepackage{amsmath,amssymb,stmaryrd,amsthm}
\theoremstyle{definition}
\newtheorem{theorem}{Theorem}[section]
\newtheorem{lemma}[theorem]{Lemma}
\newtheorem{example}[theorem]{Example}
\newtheorem{definition}[theorem]{Definition}
\newtheorem{proposition}[theorem]{Proposition}
\usepackage{enumitem}

\title{Indicative Conditionals and Dynamic Epistemic Logic}
\author{Wesley H. Holliday
\institute{University of California, Berkeley}
\email{wesholliday@berkeley.edu}
\and
Thomas F. Icard, III
\institute{Stanford University}
\email{icard@stanford.edu}
}
\def\titlerunning{Indicative Conditionals and Dynamic Epistemic Logic}
\def\authorrunning{W. H. Holliday \& T. F. Icard, III}
\begin{document}

\maketitle

\begin{abstract}
Recent ideas about epistemic modals and indicative conditionals in formal semantics have significant overlap with ideas in modal logic and dynamic epistemic logic. The purpose of this paper is to show how greater interaction between formal semantics and dynamic epistemic logic in this area can be of mutual benefit. In one direction, we show how concepts and tools from modal logic and dynamic epistemic logic can be used to give a simple, complete axiomatization of Yalcin's \cite{Yalcin2012a} semantic consequence relation for a language with epistemic modals and indicative conditionals. In the other direction, the formal semantics for indicative conditionals due to Kolodny and MacFarlane \cite{MacFarlane2010} gives rise to a new dynamic operator that is very natural from the point of view of dynamic epistemic logic, allowing succinct expression of \textit{dependence} (as in dependence logic) or \textit{supervenience} statements. We prove decidability for the logic with epistemic modals and Kolodny and MacFarlane's indicative conditional via a full and faithful computable translation from their logic to the modal logic \textsf{K45}.\end{abstract}

Logic and the formal semantics of natural language are related by blood and yet somewhat estranged. Today it is rare that formal semanticists consider questions of axiomatizability or decidability of the consequence relations defined by model-theoretic accounts of natural language fragments. Meanwhile logicians focus more on logics motivated by mathematical or philosophical concerns than on logics arising from semantic theories in linguistics.

The cost of estrangement is that insights from one field that would be useful for the other may go unnoticed or efforts may be unnecessarily duplicated. The aim of this paper is to help encourage a family reunion between logic and formal semantics of natural language, by way of concrete examples. The topic of modals and conditionals is a prime example of overlap between formal semantics and logic. In this paper, we consider the case of epistemic modals and indicative conditionals. 

In \S\ref{DELtoLanguage}, we show how concepts and tools from modal logic and dynamic epistemic logic can be used to give a simple, complete axiomatization of Yalcin's \cite{Yalcin2012a} semantic consequence relation for a language with epistemic modals and indicative conditionals. Then in \S\ref{SemanticsToDEL}, we show that the formal semantics for indicative conditionals due to Kolodny and MacFarlane \cite{MacFarlane2010} gives rise to a new dynamic operator that is very natural from the point of view of dynamic epistemic logic, allowing succinct expression of \textit{dependence} (as in dependence logic) or \textit{supervenience} statements. We prove decidability for the logic with epistemic modals and Kolodny and MacFarlane's indicative conditional via a full and faithful computable translation from their logic to the modal logic~\textsf{K45}.

There are other examples of clear overlap between formal semantics and dynamic epistemic logic, such as the connection between the \textit{dynamic logical consequence} of \cite{Willer2012} and the \textit{dynamic consequence} of  \cite{Benthem2008}, between the notions of \textit{epistemic contradictions} in \cite{Yalcin2007} and of \textit{Moorean sentences} in \cite{HI2010}, and more. Thus, the examples to follow by no means exhaust the connections to be made between formal semantics and dynamic epistemic logic. 

\section{Applying DEL to Formal Semantics}\label{DELtoLanguage}

Throughout we work with the language $\mathcal{L}(\Rightarrow)$ defined by:
\[\varphi::= p\mid \neg\varphi\mid (\varphi\wedge\varphi)\mid \Box\varphi \mid (\varphi\Rightarrow\varphi),\]
where $p$ comes from a fixed set of propositional variables. The connectives $\vee$, $\rightarrow$, $\leftrightarrow$, $\bot$, and $\Diamond$ are defined as usual, so $\Diamond\varphi :=\neg\Box\neg\varphi$. In the intended interpretation, $\Box$ and $\Diamond$ stand for ``must'' and ``might'', and $\Rightarrow$ stands for the indicative conditional ``if\dots then''. $\mathcal{L}$ is the set of formulas  that do not contain $\Rightarrow$. \textit{Nonmodal} formulas are formulas of $\mathcal{L}$ that do not contain $\Box$ (and hence $\Diamond$).

We begin by reviewing Yalcin's \cite{Yalcin2012a} semantics for $\mathcal{L}(\Rightarrow)$. Models are tuples $\mathcal{M}=\langle W,V\rangle$ where $W$ is a nonempty set and $V$ is a function assigning to each propositional variable $p$ a proposition $V(p)\subseteq W$. Formulas are evaluated in a model $\mathcal{M}$ at a world $w\in W$ relative to an \textit{information state} $X\subseteq W$ as follows:
\begin{itemize}
\item $\mathcal{M},w,X\vDash p$ iff $w\in V(p)$;
\item $\mathcal{M},w,X\vDash \neg\varphi$ iff $\mathcal{M},w,X\nvDash \varphi$; 
\item $\mathcal{M},w,X\vDash \varphi\wedge\psi$ iff $\mathcal{M},w,X\vDash \varphi$ and $\mathcal{M},w,X\vDash\psi$;
\item $\mathcal{M},w,X\vDash \Box\varphi$ iff for \textit{all} $v\in X$, $\mathcal{M},v,X\vDash \varphi$; 
\item $\mathcal{M},w,X\vDash \Diamond\varphi$ iff for \textit{some} $v\in X$, $\mathcal{M},v,X\vDash \varphi$;
\item $\mathcal{M},w,X\vDash \varphi\Rightarrow\psi$ iff $\mathcal{M},w,\llbracket \varphi\rrbracket^{\mathcal{M},X}\vDash \Box\psi$,
\end{itemize}
where $\llbracket \varphi\rrbracket^{\mathcal{M},X}=\{v\in X\mid \mathcal{M},v,X\vDash \varphi\}$.\footnote{\label{nonempty}Equivalently, we could define $\llbracket \varphi\rrbracket^{\mathcal{M},X}=\{v\in W\mid \mathcal{M},v,X\vDash \varphi\}$ and $\mathcal{M},w,X\vDash \varphi\Rightarrow\psi$ iff $\mathcal{M},w,X\cap\llbracket \varphi\rrbracket^{\mathcal{M},X}\vDash \Box\psi$, but we prefer less notation. Also note that if one wants to require  that information states be nonempty, then the clause for $\Rightarrow$ must be changed, e.g., to $\mathcal{M},w,X\vDash \varphi\Rightarrow\psi$ iff $\mathcal{M},w,X\nvDash \Diamond\varphi$ or $\mathcal{M},w,\llbracket \varphi\rrbracket^{\mathcal{M},X}\vDash \Box\psi$.} 

A formula $\varphi$ is \textit{valid} iff it is true at every world relative to every information state in every model. To define \textit{consequence}, first let $\mathcal{M},X\vDash \varphi$ ($X$ ``accepts'' $\varphi$) iff for all $w\in X$, we have $\mathcal{M},w,X\vDash\varphi$. Then for a set $\Sigma$ of formulas, Yalcin defines $\varphi$ to be an \textit{informational consequence} of $\Sigma$ iff for every model $\mathcal{M}=\langle W,V\rangle$ and information state $X\subseteq W$, if $\mathcal{M},X\vDash\sigma$ for all $\sigma\in\Sigma$, then $\mathcal{M},X\vDash \varphi$. (If $\varphi$ is valid, then $\varphi$ is an informational consequence of $\varnothing$, but the converse fails for, e.g., $p\vee\Diamond\neg p$.) 
 
The above ``domain semantics'' for $\Box$ and $\Diamond$ is presented by Yalcin \cite{Yalcin2007} as an alternative to standard Hintikka-style relational semantics for epistemic logic. Conceptually, the two semantics are different. Mathematically, the domain semantics is equivalent to a special case of the relational semantics, using what we might call \textit{uniform} relational models $\langle W,R,V\rangle$ in which every two worlds have the same set of $R$-successors. In fact, since the definition of ``informational consequence'' only involves pairs $w,X$ such that $w\in X$, it is as if we were working with relational models in which $R$ is the universal relation. In this setting, informational consequence is equivalent to the notion of \textit{global consequence} from modal logic (see \cite[\S1.5]{Blackburn2001}). Thus, from the point of view of modal logic, the above semantics for $\Box$ and $\Diamond$ can be summarized as follows: the proposed consequence relation for epistemic modals is the \textit{global consequence} relation over \textit{universal} relational models according to the standard semantics (cf. \cite{Schulz2010}).\footnote{Yalcin \cite{Yalcin2007} also proposes domain semantics for a language with both epistemic modals and attitude verbs such as `believe' and `suppose', in which case the comparison of domain semantics and relational semantics is not as straightforward.}

As for the indicative conditional $\Rightarrow$, this is where dynamic epistemic logic \cite{Ditmarsch2008,Benthem2011} enters the story. Consider a system of dynamic epistemic logic that extends the language of propositional modal logic with formulas of the form $[\varphi]\psi$, intuitively interpreted as ``after \textit{information update} with $\varphi$, $\psi$ is the case.'' Starting with relational models $\langle W,R,V\rangle$, the semantics for $[\varphi]\psi$ as in \cite[\S4.9]{Ditmarsch2008} is:
\begin{itemize} 
\item $\langle W,R,V\rangle,w\vDash [\varphi]\psi$ iff $\langle W,R^\varphi,V\rangle,w\vDash \psi$,
\end{itemize}
where $R^\varphi$ is defined by: $wR^\varphi v$ iff $wRv$ and $\langle W,R,V\rangle, v\vDash\varphi$. A superficial difference between this semantics and that of $\Rightarrow$ above is that here we are ``changing the model,'' whereas above we ``shifted the information state.''\footnote{The semantics of dynamic epistemic logic can be equivalently repackaged by pulling the relation $R$ out of the model, so we would define `$\langle W,V\rangle, w, R\vDash \varphi$', and then in the case where $R$ is uniform, by replacing $R$ with the image $X=R[W]=\{v\in W\mid \exists w\in W\colon wRv\}$, so we would define `$\langle W,V\rangle,w,X\vDash\varphi$' as in domain semantics.} The semantics for $\varphi\Rightarrow\psi$ is in fact equivalent to the semantics for $[\varphi]\Box\psi$. The form $[\varphi]\Box\psi$ has been studied extensively in dynamic epistemic logic, where the main interest is in reasoning about what is known or believed after information update, so $\Rightarrow$ turns out to be a familiar modality.

Elsewhere \cite{HI2013,HIForthcoming} we have argued that it is valuable not only in logic but also in natural language semantics to accompany a formal semantic proposal with a complete \textit{axiomatization} (when possible), capturing  basic entailment predictions of the semantics from which all of its other entailment predictions may be derived. In this spirit, and taking advantage of the connections with modal and dynamic epistemic logic above, we establish as our first main result a complete axiomatization of the logic of Yalcin's \cite{Yalcin2012a} semantics for epistemic modals and indicative conditionals.\footnote{Bledin \cite{Bledin2014} also proposes a logic for Yalcin's modal-indicative semantics, but with a Fitch-style natural deduction system for an extension of the language that includes symbols for information states and information acceptance relations.} 

\begin{theorem}\label{Ax} The set of $\mathcal{L}(\Rightarrow)$ formulas that are valid according to Yalcin's semantics is the \textnormal{Yalcin logic}: the smallest set of formulas that is closed under \textit{replacement of equivalents},\footnote{Closure under replacement of equivalents means that if $\alpha\leftrightarrow\beta\in\mathsf{L}$, and $\varphi'$ is obtained from $\varphi$ by replacing some occurrence of $\alpha$ in $\varphi$ by $\beta$, then $\varphi\leftrightarrow\varphi'\in\mathsf{L}$.} \textit{modus ponens} for the material conditional $\rightarrow$, and \textit{necessitation} for $\Box$, and contains all substitution instances of propositional tautologies and all instances of the axioms in Figure \ref{YalcinAxioms}. Moreover, $\varphi$ is an informational consequence of $\{\sigma_1,\dots,\sigma_n\}$ iff $(\Box\sigma_1\wedge\dots\wedge\Box\sigma_n)\to\Box\varphi$ is a theorem of the Yalcin logic. 
\end{theorem}

\begin{figure}[h]
\begin{center}
\begin{tabular}{ll}
\textbf{K} & $\Box(\varphi\to\psi)\to(\Box\varphi\to\Box\psi)$ \\ \\
\textbf{4} & $\Diamond\Diamond \varphi\rightarrow \Diamond\varphi$  \quad  \;\;\;\textbf{5} \quad $\Diamond\Box\varphi\rightarrow\Box\varphi$\\\\
\textbf{I1} & $(\varphi\Rightarrow \pi)\leftrightarrow \Box(\varphi\rightarrow \pi)$ for $\pi$ nonmodal\\\\
\textbf{I2} & $(\varphi\Rightarrow (\alpha\wedge\beta))\leftrightarrow ((\varphi\Rightarrow\alpha)\wedge(\varphi\Rightarrow\beta))$\\\\
\textbf{I3} & $(\varphi\Rightarrow \alpha)\to (\varphi\Rightarrow (\alpha\vee\beta))$\\\\
\textbf{I4}& $(\varphi\Rightarrow\alpha)\to (\varphi\Rightarrow \Box\alpha)$\\\\
\textbf{I5} & $((\varphi\Rightarrow (\alpha\vee \Box\beta))\wedge \neg (\varphi\Rightarrow\beta))\to (\varphi\Rightarrow \alpha)$\\\\
\textbf{I6} & $((\varphi\Rightarrow (\alpha\vee \Diamond \beta)) \wedge (\varphi\Rightarrow\neg\beta) )\to (\varphi\Rightarrow \alpha)$ \\\\
\textbf{I7} & $\neg (\varphi\Rightarrow\beta)\to (\varphi \Rightarrow \Diamond \neg \beta)$
\end{tabular}
\end{center}
\caption{Axioms of the Yalcin logic.}\label{YalcinAxioms}
\end{figure}

Before proving Theorem \ref{Ax}, let us observe how the axioms in Figure \ref{YalcinAxioms} can be seen as corresponding to natural language inferences. The following examples are taken verbatim from \cite{HIForthcoming}.

\begin{example}\label{Ex1} Axiom \textbf{I4} corresponds to a key prediction of the semantics that a sentence like
\begin{enumerate}[labelindent=\parindent, leftmargin=*, align=left,label=(\arabic*),ref=(\arabic*),resume]
\item If Miss Scarlet didn't do it, then Colonel Mustard did it.
\end{enumerate} 
entails 
\begin{enumerate}[labelindent=\parindent, leftmargin=*, align=left,label=(\arabic*),ref=(\arabic*),resume]
\item If Miss Scarlet didn't do it, then it must be that Colonel Mustard did it.
\end{enumerate}
Axiom \textbf{I5} suggests the prediction that the sentence
\begin{enumerate}[labelindent=\parindent, leftmargin=*, align=left,label=(\arabic*),ref=(\arabic*),resume]
\item If Miss Scarlet did it, then either Colonel Mustard was her accomplice or it must be that Professor Plum was involved.
\end{enumerate}
together with
\begin{enumerate}[labelindent=\parindent, leftmargin=*, align=left,label=(\arabic*),ref=(\arabic*),resume]
\item It's not the case that if Miss Scarlet did it, then Professor Plum was involved.
\end{enumerate}
entails
\begin{enumerate}[labelindent=\parindent, leftmargin=*, align=left,label=(\arabic*),ref=(\arabic*),resume]
\item If Miss Scarlet did it, then Colonel Mustard was her accomplice.
\end{enumerate}
And axiom \textbf{I6} suggests the prediction that
\begin{enumerate}[labelindent=\parindent, leftmargin=*, align=left,label=(\arabic*),ref=(\arabic*),resume]
\item If Miss Scarlet did it, then either she used the pipe or she might have used the candlestick.
\end{enumerate}
together with
\begin{enumerate}[labelindent=\parindent, leftmargin=*, align=left,label=(\arabic*),ref=(\arabic*),resume]
\item If Miss Scarlet did it, she didn't use the candlestick.
\end{enumerate}
entails
\begin{enumerate}[labelindent=\parindent, leftmargin=*, align=left,label=(\arabic*),ref=(\arabic*),resume]
\item If Miss Scarlet did it, then she used the pipe.
\end{enumerate}
Finally, axiom \textbf{I7} suggests the prediction that one who \textit{rejects}
\begin{enumerate}[labelindent=\parindent, leftmargin=*, align=left,label=(\arabic*),ref=(\arabic*),resume]
\item If Miss Scarlet was in the ballroom, then Colonel Mustard is guilty.
\end{enumerate}
should accept
\begin{enumerate}[labelindent=\parindent, leftmargin=*, align=left,label=(\arabic*),ref=(\arabic*),resume]
\item If Miss Scarlet was in the ballroom, it might be that Colonel Mustard is not guilty.
\end{enumerate}
\end{example}

We will prove Theorem \ref{Ax} using several lemmas. The first lemma follows from the well-known fact that \textsf{K45} is the logic of the ``uniform'' relational models mentioned above, plus the equivalence of domain semantics and uniform relational semantics for $\mathcal{L}$.

\begin{lemma}\label{K45comp} The set of $\mathcal{L}$ formulas that are valid according to Yalcin's semantics is the logic \textsf{K45}.
\end{lemma}

The next two lemmas show that the conditional axioms of the Yalcin logic, which we have presented in their form in Figure \ref{YalcinAxioms} to bring out connections with natural language inference as in Example \ref{Ex1}, can be equivalently repackaged as valid \textit{reduction axioms} in the sense of dynamic epistemic logic \cite{Ditmarsch2008}.

\begin{lemma}\label{ValidAx} The following axioms are valid according to Yalcin's semantics:
\begin{enumerate}
\item[\textbf{A1}] $(\varphi\Rightarrow \pi)\leftrightarrow \Box(\varphi\rightarrow \pi)$ for $\pi$ a nonmodal formula;
\item[\textbf{A2}] $(\varphi\Rightarrow (\alpha\wedge\beta))\leftrightarrow ((\varphi\Rightarrow\alpha)\wedge(\varphi\Rightarrow\beta))$;
\item[\textbf{A3}] $(\varphi\Rightarrow (\alpha\vee \Box\beta))\leftrightarrow ((\varphi\Rightarrow \alpha)\vee (\varphi\Rightarrow\beta))$;
\item[\textbf{A4}] $(\varphi\Rightarrow (\alpha\vee \Diamond \beta))\leftrightarrow ((\varphi\Rightarrow \alpha)\vee \neg (\varphi\Rightarrow\neg\beta))$.
\end{enumerate}
\end{lemma}

\begin{proof} For \textbf{A1}, suppose $\pi$ is nonmodal. By definition, we have
\begin{eqnarray}
\mathcal{M},w,X\vDash \varphi\Rightarrow\pi & \mbox{iff} & \mathcal{M},w,\llbracket \varphi\rrbracket^{\mathcal{M},X}\vDash \Box\pi \nonumber\\
 & \mbox{iff} & \llbracket \varphi\rrbracket^{\mathcal{M},X} \subseteq \llbracket \pi \rrbracket^{\mathcal{M},\llbracket \varphi\rrbracket^{\mathcal{M},X}}.\label{1}
\end{eqnarray}
Since $\pi$ is nonmodal, we have
\begin{eqnarray*}
\llbracket \pi \rrbracket^{\mathcal{M},\llbracket \varphi\rrbracket^{\mathcal{M},X}} = \llbracket \varphi\rrbracket^{\mathcal{M},X} \cap \llbracket \pi \rrbracket^{\mathcal{M},X},
\end{eqnarray*}
so the right side of (\ref{1}) is equivalent to
\begin{eqnarray*}\llbracket \varphi\rrbracket^{\mathcal{M},X}\subseteq \llbracket \pi \rrbracket^{\mathcal{M},X}, \end{eqnarray*}
which by definition is equivalent to $\mathcal{M},w,X\vDash \Box (\varphi\to\pi)$.

For \textbf{A2}, by definition we have
\begin{eqnarray*}
\mathcal{M},w,X\vDash \varphi\Rightarrow (\alpha_1\wedge \alpha_2) & \mbox{iff} & \mathcal{M},w,\llbracket \varphi\rrbracket^{\mathcal{M}, X}\vDash\Box (\alpha_1\wedge\alpha_2) \\
& \mbox{iff} & \mathcal{M},w,\llbracket \varphi\rrbracket^{\mathcal{M}, X}\vDash\Box \alpha_1\wedge \Box\alpha_2 \\
&\mbox{iff} & \mathcal{M},w,\llbracket \varphi\rrbracket^{\mathcal{M}, X}\vDash \Box\alpha_i,\; i=1,2 \\
&\mbox{iff} & \mathcal{M},w,X\vDash \varphi\Rightarrow \alpha_i,\;i=1,2.
\end{eqnarray*}

For \textbf{A3}, by definition we have
\begin{eqnarray} \mathcal{M},w,X\vDash \varphi\Rightarrow (\alpha\vee \Box\beta)
& \mbox{iff} & \mathcal{M},w,\llbracket \varphi\rrbracket^{\mathcal{M}, X}\vDash\Box (\alpha\vee\Box\beta)  \nonumber\\
&\mbox{iff}& \llbracket \varphi\rrbracket^{\mathcal{M}, X}\subseteq \llbracket \alpha\vee\Box\beta\rrbracket^{\mathcal{M},\llbracket \varphi\rrbracket^{\mathcal{M}, X}}\nonumber\\
&\mbox{iff}& \llbracket \varphi\rrbracket^{\mathcal{M}, X}\subseteq \llbracket \alpha\rrbracket^{\mathcal{M},\llbracket \varphi\rrbracket^{\mathcal{M}, X}}\cup\llbracket\Box\beta\rrbracket^{\mathcal{M},\llbracket \varphi\rrbracket^{\mathcal{M}, X}}.\label{2}
\end{eqnarray}
Using (\ref{2}), we first show that if $\mathcal{M},w,X\vDash \varphi\Rightarrow (\alpha\vee \Box\beta)$, then $\mathcal{M},w,X\vDash \varphi\Rightarrow \alpha$ or $\mathcal{M},w,X\vDash \varphi\Rightarrow \beta$. 

Case 1: $\llbracket\Box\beta\rrbracket^{\mathcal{M},\llbracket \varphi\rrbracket^{\mathcal{M}, X}}=\varnothing$. Then (\ref{2}) implies
\begin{eqnarray}
\llbracket \varphi\rrbracket^{\mathcal{M}, X}\subseteq \llbracket \alpha\rrbracket^{\mathcal{M},\llbracket \varphi\rrbracket^{\mathcal{M}, X}} 
&\Rrightarrow&\mathcal{M},w,X\vDash \varphi\Rightarrow\alpha.\label{3}
\end{eqnarray}

Case 2: $\llbracket\Box\beta\rrbracket^{\mathcal{M},\llbracket \varphi\rrbracket^{\mathcal{M}, X}}\neq\varnothing$. Then 
\[\exists v\in W\colon \mathcal{M}, v, \llbracket \varphi\rrbracket^{\mathcal{M}, X}\vDash \Box\beta ,\]
and we have
\begin{eqnarray*}
\mathcal{M}, v, \llbracket \varphi\rrbracket^{\mathcal{M}, X}\vDash \Box\beta 
&\mbox{iff}& \llbracket \varphi\rrbracket^{\mathcal{M}, X}\subseteq \llbracket \beta\rrbracket^{\mathcal{M},\llbracket \varphi\rrbracket^{\mathcal{M}, X}}\\
&\mbox{iff}&\mathcal{M},w,X\vDash \varphi\Rightarrow \beta.
\end{eqnarray*}

Next, we show that if $\mathcal{M},w,X\vDash \varphi\Rightarrow \alpha$ or $\mathcal{M},w,X\vDash \varphi\Rightarrow \beta$, then $\mathcal{M},w,X\vDash \varphi\Rightarrow (\alpha\vee \Box\beta)$, using (\ref{2}). If $\mathcal{M},w,X\vDash \varphi\Rightarrow\alpha$, then $\llbracket \varphi\rrbracket^{\mathcal{M}, X}\subseteq \llbracket \alpha\rrbracket^{\mathcal{M},\llbracket \varphi\rrbracket^{\mathcal{M}, X}} $, which implies (\ref{2}). On the other hand, if $\mathcal{M},w,X\vDash \varphi\Rightarrow \beta$, then we have
\begin{eqnarray*}
 \llbracket \varphi\rrbracket^{\mathcal{M}, X}\subseteq \llbracket \beta\rrbracket^{\mathcal{M},\llbracket \varphi\rrbracket^{\mathcal{M}, X}} 
&\Rrightarrow &\forall v\in X\colon  \mathcal{M}, v, \llbracket \varphi\rrbracket^{\mathcal{M}, X}\vDash \Box\beta\\
&\Rrightarrow& \forall v\in \llbracket \varphi\rrbracket^{\mathcal{M}, X}\colon \mathcal{M}, v, \llbracket \varphi\rrbracket^{\mathcal{M}, X}\vDash \Box\beta \\
&\Rrightarrow &  \llbracket \varphi\rrbracket^{\mathcal{M}, X}\subseteq \llbracket \Box \beta\rrbracket^{\mathcal{M},\llbracket \varphi\rrbracket^{\mathcal{M}, X}} ,
\end{eqnarray*}
which implies (\ref{2}).

For \textbf{A4}, by similar reasoning to that for (\ref{2}), we have
\begin{eqnarray}
 \mathcal{M},w,X\vDash \varphi\Rightarrow (\alpha\vee \Diamond\beta) 
&\mbox{iff}& \llbracket \varphi\rrbracket^{\mathcal{M}, X}\subseteq \llbracket \alpha\rrbracket^{\mathcal{M},\llbracket \varphi\rrbracket^{\mathcal{M}, X}}\cup\llbracket\Diamond\beta\rrbracket^{\mathcal{M},\llbracket \varphi\rrbracket^{\mathcal{M}, X}}.\label{4}
\end{eqnarray}
Using (\ref{4}), we show that if $ \mathcal{M},w,X\vDash \varphi\Rightarrow (\alpha\vee \Diamond\beta)$, then $ \mathcal{M},w,X\vDash \varphi\Rightarrow\alpha$ or $\mathcal{M},w,X\vDash \neg (\varphi\Rightarrow\neg\beta)$.

Case 1:  $\llbracket\Diamond\beta\rrbracket^{\mathcal{M},\llbracket \varphi\rrbracket^{\mathcal{M}, X}}=\varnothing$. Then as in Case 1 for \textbf{A3}, (\ref{4}) implies $\mathcal{M},w,X\vDash \varphi\Rightarrow\alpha$.

Case 2: $\llbracket\Diamond\beta\rrbracket^{\mathcal{M},\llbracket \varphi\rrbracket^{\mathcal{M}, X}}\neq\varnothing$. Then we observe that
\begin{eqnarray}
\llbracket\Diamond\beta\rrbracket^{\mathcal{M},\llbracket \varphi\rrbracket^{\mathcal{M}, X}}\neq\varnothing
&\mbox{iff} &\exists v\in \llbracket \varphi\rrbracket^{\mathcal{M}, X} \colon  \mathcal{M}, v, \llbracket \varphi\rrbracket^{\mathcal{M}, X}\vDash \Diamond\beta\nonumber\\
&\mbox{iff}& \llbracket \varphi\rrbracket^{\mathcal{M}, X}\cap \llbracket \beta\rrbracket^{\mathcal{M},\llbracket \varphi\rrbracket^{\mathcal{M}, X}}\neq \varnothing \nonumber \\
&\mbox{iff}& \llbracket \varphi\rrbracket^{\mathcal{M}, X}\not\subseteq W\setminus \llbracket \beta\rrbracket^{\mathcal{M},\llbracket \varphi\rrbracket^{\mathcal{M}, X}} \nonumber \\
&\mbox{iff}& \llbracket \varphi\rrbracket^{\mathcal{M}, X}\not\subseteq \llbracket \neg\beta\rrbracket^{\mathcal{M},\llbracket \varphi\rrbracket^{\mathcal{M}, X}} \nonumber \\
&\mbox{iff}& \mathcal{M},w,\llbracket \varphi\rrbracket^{\mathcal{M}, X}\nvDash \Box\neg\beta\nonumber \\
&\mbox{iff}& \mathcal{M},w,X\nvDash \varphi\Rightarrow\neg\beta\nonumber \\
&\mbox{iff}& \mathcal{M},w,X\vDash \neg(\varphi\Rightarrow\neg\beta).\label{5}
\end{eqnarray}

Finally, we show that if $\mathcal{M},w,X\vDash \varphi\Rightarrow \alpha$ or $\mathcal{M},w,X\vDash \neg(\varphi\Rightarrow\neg\beta)$, then $\mathcal{M},w,X\vDash \varphi\Rightarrow (\alpha\vee \Box\beta)$, using (\ref{4}). The argument starting from $\mathcal{M},w,X\vDash \varphi\Rightarrow \alpha$ is the same as for \textbf{A3}. The argument from $\mathcal{M},w,X\vDash \neg(\varphi\Rightarrow\neg\beta)$  uses (\ref{5}) and the observation that
\begin{eqnarray*}
\llbracket \varphi\rrbracket^{\mathcal{M}, X}\cap \llbracket \beta\rrbracket^{\mathcal{M},\llbracket \varphi\rrbracket^{\mathcal{M}, X}}\neq \varnothing
&\Rrightarrow &\forall v\in X\colon  \mathcal{M}, v, \llbracket \varphi\rrbracket^{\mathcal{M}, X}\vDash \Diamond\beta\\
&\Rrightarrow& \forall v\in \llbracket \varphi\rrbracket^{\mathcal{M}, X}\colon \mathcal{M}, v, \llbracket \varphi\rrbracket^{\mathcal{M}, X}\vDash \Diamond\beta \\
&\Rrightarrow &  \llbracket \varphi\rrbracket^{\mathcal{M}, X}\subseteq \llbracket \Diamond \beta\rrbracket^{\mathcal{M},\llbracket \varphi\rrbracket^{\mathcal{M}, X}} ,
\end{eqnarray*}
which implies (\ref{4}).
\end{proof}

Next we verify that the axioms of Lemma \ref{ValidAx} are indeed an equivalent repackaging of the axioms of the Yalcin logic.

\begin{lemma}\label{EquivAx}
A formula $\varphi$ is a theorem of the Yalcin logic iff $\varphi$ is a theorem of the logic defined in the same way as the Yalcin logic (in Theorem \ref{Ax}) but with axioms \textbf{A1}-\textbf{A4} in place of \textbf{I1}-\textbf{I7}. 
\end{lemma}

\begin{proof}Axioms \textbf{I1} and \textbf{I2} are just axioms \textbf{A1} and \textbf{A2}, so we begin with \textbf{I3}. As an instance of \textbf{A2}, $(\varphi\Rightarrow (\delta\wedge\gamma))\leftrightarrow ((\varphi\Rightarrow\delta)\wedge(\varphi\Rightarrow\gamma))$,  we have
\[(\varphi\Rightarrow ((\alpha\vee\beta)\wedge (\alpha\vee\neg\beta)))\leftrightarrow ((\varphi\Rightarrow (\alpha\vee\beta))\wedge (\varphi\Rightarrow (\alpha\vee\neg\beta))).\]
Since $(\alpha\vee\beta)\wedge (\alpha\vee\neg\beta)$ is equivalent to $\alpha$, by replacement of equivalents the left-to-right direction of the biconditional gives us $(\varphi\Rightarrow \alpha)\to (\varphi\Rightarrow (\alpha\vee\beta))$, which is \textbf{I3}.

For \textbf{I4},  as an instance of \textbf{A3}, $(\varphi\Rightarrow (\alpha\vee \Box\beta))\leftrightarrow ((\varphi\Rightarrow \alpha)\vee (\varphi\Rightarrow\beta))$, we have 
\[(\varphi\Rightarrow (\bot\vee\Box\beta))\leftrightarrow ((\varphi\Rightarrow \bot)\vee (\varphi\Rightarrow\beta)).\]
Since $\bot\vee\Box\beta$ is equivalent to $\Box\beta$, by replacement of equivalents the right-to-left direction of the biconditional gives us $(\varphi\Rightarrow\beta)\to (\varphi\Rightarrow \Box\beta)$, which is \textbf{I4}.

 \textbf{I5} is a Boolean rewriting of the left-to-right direction of \textbf{A3}, and \textbf{I6} is a Boolean rewriting of the left-to-right direction of \textbf{A4}.

Finally, for \textbf{I7}, as an instance of \textbf{A4}, $(\varphi\Rightarrow (\delta\vee \Diamond \gamma))\leftrightarrow ((\varphi\Rightarrow \delta)\vee \neg (\varphi\Rightarrow\neg\gamma))$, we have
\[(\varphi\Rightarrow (\bot\vee\Diamond\neg\beta))\leftrightarrow ((\varphi\Rightarrow\bot)\vee\neg(\varphi\Rightarrow\neg\neg\beta)),\]
which by replacement of equivalents is equivalent to 
\[(\varphi\Rightarrow \Diamond\neg\beta)\leftrightarrow ((\varphi\Rightarrow\bot)\vee\neg(\varphi\Rightarrow\beta)),\]
the right-to-left direction of which gives us $\neg (\varphi\Rightarrow\beta)\to (\varphi \Rightarrow \Diamond \neg \beta)$, which is \textbf{I7}.

In the other direction, given the observations above, it only remains to show that the right-to-left directions of \textbf{A3} and \textbf{A4} are theorems of the Yalcin logic, which by Boolean reasoning reduces to showing that the following four formulas are theorems of the Yalcin logic:
\begin{eqnarray}
&&(\varphi\Rightarrow\alpha)\rightarrow (\varphi\Rightarrow (\alpha\vee\Box\beta))\label{1of4}\\
&&(\varphi\Rightarrow\beta)\rightarrow (\varphi\Rightarrow (\alpha\vee\Box\beta))\label{2of4} \\
&&(\varphi\Rightarrow\alpha)\rightarrow (\varphi\Rightarrow (\alpha\vee\Diamond\beta))\label{3of4}\\
&&\neg(\varphi\Rightarrow\neg\beta)\rightarrow (\varphi\Rightarrow (\alpha\vee\Diamond\beta)).\label{4of4}
\end{eqnarray}
Formula (\ref{1of4}) is an instance of \textbf{I3}. For (\ref{2of4}), as instances of \textbf{I4} and \textbf{I3}, we have:
\begin{eqnarray*}
&&(\varphi\Rightarrow\beta)\rightarrow(\varphi\Rightarrow\Box\beta)\\
&&(\varphi\Rightarrow\Box\beta)\rightarrow(\varphi\Rightarrow(\Box\beta\vee\alpha)),
\end{eqnarray*}
and we can use replacement of equivalents to replace $\Box\beta\vee\alpha$ by $\alpha\vee\Box\beta$. Next, (\ref{3of4}) is an instance of \textbf{I3}. Finally, for (\ref{4of4}), as an instance of \textbf{I7}, we have
\[\neg(\varphi\Rightarrow\neg\beta)\rightarrow(\varphi\Rightarrow \Diamond\neg\neg\beta)\]
and hence
\[\neg(\varphi\Rightarrow\neg\beta)\rightarrow(\varphi\Rightarrow \Diamond\beta)\]
by replacement of equivalents; and as an instance of \textbf{I3}, we have
\[(\varphi\Rightarrow \Diamond\beta)\rightarrow (\varphi\Rightarrow (\Diamond\beta\vee\alpha))\]
and hence 
\[(\varphi\Rightarrow \Diamond\beta)\rightarrow (\varphi\Rightarrow (\alpha\vee\Diamond\beta))\]
by replacement of equivalents. Putting the above implications together, we can derive (\ref{4of4}).
\end{proof}

For the next step in our argument, we use the following well-known fact about the modal logic \textsf{K45}, which is exactly the set of theorems of the Yalcin logic restricted to the language $\mathcal{L}$.

\begin{lemma}\label{NormalForm} Every formula $\varphi\in\mathcal{L}$ is provably equivalent in \textsf{K45} to a formula of the form 
\[\underset{1\leq i\leq n}{\bigwedge} (\pi^i\vee \Diamond\beta^i\vee \Box\beta^i_1\vee\dots\vee \Box \beta^i_{m_i}),\]
where $\pi^i$, $\beta^i$, and $\beta^i_1,\dots,\beta^i_{m_i}$ are nonmodal formulas, and to a formula of the form 
\[\underset{1\leq i\leq n}{\bigvee} (\pi^i\wedge \Box\beta^i\wedge \Diamond\beta^i_1\wedge\dots\wedge \Diamond \beta^i_{m_i}),\]
where $\pi^i$, $\beta^i$, and $\beta^i_1,\dots,\beta^i_{m_i}$ are nonmodal formulas.
\end{lemma}

Now we take advantage of Lemmas \ref{EquivAx} and \ref{NormalForm} to show that using the Yalcin logic, every formula containing conditionals  can be reduced to an equivalent formula without conditionals. 

\begin{lemma}\label{Reduction} Every formula of $\mathcal{L}(\Rightarrow)$ is provably equivalent in the Yalcin logic to a formula of $\mathcal{L}$.
\end{lemma}
\begin{proof} It suffices to show that any conditional formula \textit{containing no embedded conditionals} is equivalent to a formula of $\mathcal{L}$.  For in any formula containing conditionals, there must be ``innermost'' occurrences of conditional formulas containing no embedded conditionals, and then we can replace those occurrences of conditional formulas---using replacement of equivalents---with equivalent formulas of $\mathcal{L}$, repeating this process until we eventually obtain a formula containing no conditionals. For a rigorous treatment of such a reduction argument, see \cite[\S7.4]{Ditmarsch2008}.

In what follows, we use the fact from Lemma \ref{EquivAx} that \textbf{A1}--\textbf{A4} are derivable in the Yalcin logic.

Given a conditional formula $\varphi\Rightarrow\psi$ with no conditionals in $\varphi$ or $\psi$, we use the fact from Lemma \ref{NormalForm} that $\psi$ is equivalent to a formula $\psi'$ in \textsf{K45} conjunctive normal form:
\[\psi':=\underset{1\leq i\leq n}{\bigwedge} (\pi^i\vee \Diamond\beta^i\vee \Box \beta^i_1\vee\dots\vee \Box\beta_{m_i}^i).\]
By replacement of equivalents, $\varphi\Rightarrow\psi$ is equivalent to $\varphi\Rightarrow\psi'$. By repeated application of \textbf{A2}, $\varphi\Rightarrow\psi'$ is equivalent to
\[\underset{1\leq i\leq n}{\bigwedge}(\varphi\Rightarrow (\pi^i\vee \Diamond\beta^i\vee \Box \beta^i_1\vee\dots\vee \Box\beta_{m_i}^i)).\]
So it suffices to show that each formula of the form
\[\varphi\Rightarrow (\pi\vee \Diamond\beta\vee \Box \beta_1\vee\dots\vee \Box\beta_m)\]
is equivalent to a formula of $\mathcal{L}$. Let
\[\alpha:=\pi\vee \Box \beta_1\vee\dots\vee \Box\beta_m. \]
By replacement of equivalents, we can equivalently write the above conditional as 
\[\varphi\Rightarrow (\alpha\vee\Diamond\beta),\]
which is equivalent, by \textbf{A4}, to
\[(\varphi\Rightarrow \alpha)\vee \neg(\varphi\Rightarrow\neg\beta),\]
which is in turn equivalent, by \textbf{A1}, to
\[(\varphi\Rightarrow \alpha)\wedge \neg\Box(\varphi\to \neg\beta),\]
since $\beta$ is nonmodal. Now it suffices to show that each formula of the form 
\[\varphi\Rightarrow (\pi\vee \Box \beta_1\vee\dots\vee \Box\beta_m) \]
is equivalent to a formula of $\mathcal{L}$. We do so by induction on $m$. Now let 
\[\alpha:= \pi\vee\Box\beta_1\vee\dots\vee\Box\beta_{m-1},\]
so we can write the above conditional as 
\[\varphi\Rightarrow (\alpha\vee\Box\beta_m),\]
which is equivalent, by \textbf{A3}, to 
\[(\varphi\Rightarrow\alpha)\wedge (\varphi\Rightarrow\beta_m),\]
which is in turn equivalent, by \textbf{A1}, to
\[(\varphi\Rightarrow\alpha)\wedge \Box(\varphi\to\beta_m),\]
because $\beta_m$ is nonmodal. By the inductive hypothesis, the first disjunct is equivalent to a formula of $\mathcal{L}$. This completes the proof that our original formula $\varphi\Rightarrow\psi$ is equivalent to a formula of $\mathcal{L}$.
\end{proof}

We can now put everything together to prove Theorem \ref{Ax}.

\begin{proof} For soundness, if $\varphi$ is derivable in the Yalcin logic, then by Lemma \ref{EquivAx}, $\varphi$ is derivable in the logic defined with the axioms of Lemma \ref{ValidAx}, which we have shown to be valid. Since the rules also preserve validity, $\varphi$ is valid. For completeness, if $\varphi$ is valid, then so is its equivalent $\varphi'$ in the modal language $\mathcal{L}$ provided by Lemma \ref{Reduction}, given soundness. By the completeness of \textsf{K45} in Lemma \ref{K45comp}, $\varphi'$ is provable in \textsf{K45} and hence in the Yalcin logic, and by Lemma \ref{Reduction}, the Yalcin logic proves $\varphi\leftrightarrow\varphi'$. Thus, the Yalcin logic proves $\varphi$. Finally, it is easy to see that $\varphi$ is an informational consequence of $\{\sigma_1,\dots,\sigma_n\}$ iff $(\Box\sigma_1\wedge\dots\wedge\Box\sigma_n)\to\Box\varphi$ is valid according to Yalcin's semantics and hence iff $(\Box\sigma_1\wedge\dots\wedge\Box\sigma_n)\to\Box\varphi$ is a theorem of the Yalcin logic by our previous reasoning. \end{proof}

\section{From Formal Semantics to DEL}\label{SemanticsToDEL}

In this section, we consider an alternative semantics for the indicative conditional, suggested in the formal semantics literature. This semantics is also of interest purely from the perspective of dynamic epistemic logic, but it has not been previously considered in the DEL literature.

The semantics for the indicative $\Rightarrow$ in \S\ref{DELtoLanguage} was intended for the case where the antecedent is \textit{nonmodal}. If the antecedent is modal, things become trickier, for the well-known reason in dynamic epistemic logic \cite{Ditmarsch2008,HI2010} that updating with the antecedent $\varphi$ may fail to bring about an information state that accepts $\varphi$. The desire is to update the information state in such a way that the antecedent $\varphi$ is accepted and then check whether the consequent $\psi$ is accepted. However, we cannot say that the information state $X$ should be updated to the largest $X'\subseteq X$ that accepts the antecedent $\varphi$, because even if there is such a subset, there might fail to be a unique \textit{largest} one. A solution from Kolodny and MacFarlane \cite{MacFarlane2010} is the following:
\begin{itemize}
\item $\mathcal{M},w,X\vDash \varphi\Rightarrow\psi$ iff $\mathcal{M},w,X'\vDash \Box\psi$ for all $X'$ such that \\ (i) $X'\subseteq X$, (ii) $X'\subseteq\llbracket \varphi\rrbracket^{\mathcal{M},X'}$, and (iii) there is no $X''$ satisfying (i) and (ii) such that $X'\subsetneq X''$.
\end{itemize}

Although this semantics is equivalent to that of \S\ref{DELtoLanguage} for nonmodal $\varphi$, they are not equivalent in general.

\begin{example}\label{Ex2} Consider a two-world model $\mathcal{M}$ with $w$ and $v$ such that $p$ is true only at $w$ and $q$ is true only at $v$. According to the semantics of \S\ref{DELtoLanguage}, the formula $(\Box p\vee \Box \neg p)\Rightarrow  q$ is trivially true at $\mathcal{M},w,\{w,v\}$, because $\llbracket\Box p\vee\Box\neg p\rrbracket^{\mathcal{M},\{w,v\}}=\varnothing$. By contrast, according to the Kolodny-MacFarlane semantics,  the formula $(\Box p\vee \Box \neg p)\Rightarrow  q$ is false at $\mathcal{M},w,\{w,v\}$, because $\{w\}$ is an $X'\subseteq \{w,v\}$ satisfying conditions (i), (ii), and (iii), and yet $\mathcal{M},w,\{w\}\nvDash \Box q$.
\end{example}

\begin{example}\label{Ex3} According to the Kolodny-MacFarlane semantics, $(p\wedge\Diamond \neg p)\Rightarrow\bot$ is valid, because there can be no nonempty $X'$ such that $X'\subseteq\llbracket p\wedge\Diamond\neg p\rrbracket^{\mathcal{M},X'}$. By contrast, according to the semantics of \S\ref{DELtoLanguage}, $(p\wedge\Diamond\neg p)\Rightarrow\bot$ is invalid, because on that semantics $(p\wedge\Diamond\neg p)\Rightarrow\varphi$ is equivalent to $p\Rightarrow\varphi$ for any $\varphi$.
\end{example}

From the point of view of dynamic epistemic logic, the above semantics can be seen as interpreting $\varphi\Rightarrow \psi$ as a very natural statement: ``every minimal epistemic change yielding knowledge of $\varphi$ also yields knowledge of $\psi$.'' As Example \ref{Ex2} shows, there may be more than one minimal epistemic change yielding knowledge of $\varphi$, and as Example \ref{Ex3} shows, there may be no epistemic change yielding knowledge of $\varphi$. The standard interpretation of $[\varphi]\Box\psi$ in dynamic epistemic logic, as ``information update with $\varphi$ yields knowledge of $\psi$,'' is quite different, due to the phenomenon of \textit{unsuccessful update} alluded to above, wherein updating with $\varphi$ fails to produce an epistemic state in which $\varphi$ is known. The difference between $[\varphi]\Box\psi$ under the standard semantics and $\varphi\Rightarrow\psi$ under the Kolodny-MacFarlane semantics is that in the first case, $\varphi$ specifies the \textit{mechanism of epistemic change}---remove from the information state all worlds that  satisfied $\neg\varphi$---whereas in the second case, $\varphi$ specifies the \textit{desired result of epistemic change}---an epistemic state in which $\varphi$ is known. We think that both ways of reasoning are important for dynamic epistemic logicians interested in information update and learning. There are also natural generalizations of the Kolodny-MacFarlane semantics to the multi-agent setting in which dynamic epistemic logicians typically work (e.g., every minimal epistemic change yielding common knowledge of $\varphi$ also yields common knowledge of $\psi$), but in this preliminary analysis we restrict attention to the single-agent case.

Among the interesting application of the Kolodny-MacFarlane conditional is to succinctly express \textit{dependence} as in modal dependence logic \cite{Vaananen2008} or \textit{supervenience} as in philosophy \cite{McLaughlin2014}.

\begin{definition} In an information state $X$, the (truth value of the) propositional variable $q$ \textit{depends on} or \textit{supervenes on} the (truth values of the) propositional variables $p_1,\dots,p_n$ iff any two worlds in $X$ that agree on the truth values of $p_1,\dots,p_n$ also agree on the truth value of $q$. 
\end{definition}

The dependence of $q$ on $p_1,\dots,p_n$ can be expressed in the language $\mathcal{L}$ by the formula 
\begin{equation}\underset{s\in \mathsf{state}(p_1,\dots,p_n)}{\bigwedge} (\Box (s\to q) \vee \Box (s\to\neg q)),\label{Expo}\end{equation} where $\mathsf{state}(p_1,\dots,p_n)$ is the set of all conjunctions of the form $\pm_1 p_1\wedge\dots\wedge \pm_n p_n$ with $\pm_i$ being either $\neg$ or empty. But whereas the formula (\ref{Expo}) grows exponentially in the size of $n$, the Kolodny-MacFarlane conditional allows us to express dependence with a formula that grows only linearly in the size of $n$.

\begin{proposition} In an information state $X$, the propositional variable $q$ depends on $p_1,\dots,p_n$ iff the following formula is true relative to $X$ according to the Kolodny-MacFarlane semantics for $\Rightarrow$:
\begin{equation}(\underset{1\leq i\leq n}{\bigwedge}(\Box p_i\vee \Box \neg p_i) )\Rightarrow (\Box q\vee\Box\neg q).\label{Depend}\end{equation}
\end{proposition}
\begin{proof}
It is easy to see that if $q$ depends on $p_1,\dots,p_n$ in $X$, then (\ref{Depend}) is true relative to $X$. Conversely, suppose (\ref{Depend}) is true relative to $X$, and $x$ and $y$ are worlds in $X$ that agree on the truth values of  $p_1,\dots,p_n$. Then we claim that the set $X'$ of all worlds in $X$ that agree with $x$ on the truth values of $p_1,\dots,p_n$ meets conditions (i), (ii), and (iii) above. For (ii), by the definition of $X'$, we have $\mathcal{M},w,X'\vDash \Box p_i\vee\Box\neg p_i$ for each $i$. For (iii), note that if $X'\subsetneq X''\subseteq X$, then there is a $z\in X''\setminus X'$ that disagrees with $x$ on the truth value of some  $p_i$, which with $x,z\in X''$ implies that $\mathcal{M},w,X''\nvDash \Box p_i\vee\Box\neg p_i$, so $X''$ does not satisfy (ii). Now since $X'$ satisfies (i), (ii), and (iii), and (\ref{Depend}) is true relative to $X$, it follows that $\mathcal{M},w,X'\vDash \Box q\vee\Box\neg q$, which implies that all worlds in $X'$, and in particular $x$ and $y$, agree on the truth value of $q$.
\end{proof}

Below we will provide a computable translation $(\cdot)^\dagger$ from the language $\mathcal{L}(\Rightarrow)$ into the basic modal language $\mathcal{L}$. As shown in Lemma \ref{Lem1}, this translation preserves (in)validity according to Kolodny and MacFarlane's semantics. By Lemma \ref{K45comp}
 and the fact that we have not changed the semantics from \S\ref{DELtoLanguage} for formulas without conditionals, this provides a full and faithful translation from the logic with epistemic modals and Kolodny and MacFarlane's indicative conditional to the logic \textsf{K45}. As is well-known, \textsf{K45} is decidable, so it follows that Kolodny and MacFarlane's logic is also decidable.

The strategy for our translation is as follows. We first define the translation $\lambda ^*$ for a conditional formula $\lambda$ with no embedded conditionals, using \textsf{K45} normal forms (Lemma \ref{NormalForm}). We then extend this to a translation from the full language by induction.

Suppose we are given a conditional formula
\[\lambda := \Theta \Rightarrow \Omega \]
where $\Theta$ and $\Omega$ are in \textsf{K45} \emph{disjunctive} normal form, so that \[\Theta:=\underset{i\in I}\bigvee\theta_i\quad\mbox{ and }\quad\Omega:=\underset{j\in J}\bigvee\omega_j,\] where  for nonmodal formulas $\varphi_i$, $\psi_i$, $\chi_n$, $\alpha_j$, $\beta_j$, and $\gamma_m$,
\[ \theta_i:= \varphi_i\wedge \Box\psi_i \wedge \underset{n\in D_i}{\bigwedge} \Diamond\chi_n \quad\mbox{ and }\quad \omega_j:= \alpha_j\wedge \Box\beta_j \wedge \underset{m\in D_j}{\bigwedge} \Diamond\gamma_m.\]
We would like our translation $\lambda^*$ to express in the basic modal language that every maximal set of worlds making some of the $\theta_i$ formulas true also makes at least one of the $\omega_j$ formulas true.  Where $K$ indexes some subset of the $\theta_i$ formulas,  $\mathtt{info}_K$ below gives the nonmodal information that each world in an information state must satisfy in order for the information state to accept $\bigvee_{i \in K} \theta_i$, while $\mathtt{good}_K$ asserts that $\bigvee_{i \in K} \theta_i$ will indeed be accepted when we restrict the current information state to the set of worlds satisfying $\mathtt{info}_K$, thanks to sufficient witnesses for the $\Diamond \chi_n$ formulas. Meanwhile, $\mathtt{max}_K$ adds that $K$ is maximal, with respect to the set of worlds satisfying $\mathtt{info}_K$, among such subsets. For $K\subseteq I$, let 
\[\mathtt{info}_K:= (\underset{k\in K}{\bigvee} \varphi_k )\wedge \underset{k\in K}{\bigwedge}\psi_k,\] 
let 
\[\mathtt{good}_K:= \underset{k\in K}{\bigwedge} \,\underset{n \in D_k}{\bigwedge} \Diamond (\mathtt{info}_K \wedge\chi_n),\]
and let
\begin{eqnarray*}
&&\mathtt{max}_K:= \mathtt{good}_K\wedge \underset{L\subseteq I}{\bigwedge}\Big(\big(\Box(\mathtt{info}_K\rightarrow \mathtt{info}_L)\wedge  \Diamond (\neg \mathtt{info}_K\wedge\mathtt{info}_L)\big)\rightarrow \neg\mathtt{good}_L \Big) .\end{eqnarray*}Concerning the consequent formula $\Omega$: for $S\subseteq J$, let 
\[\mathtt{state}_S:= \underset{s\in S}{\bigwedge}\alpha_s\wedge \underset{s\in J\setminus S}{\bigwedge}\neg\alpha_s.\]
Finally, given our $\lambda$ above, we define
\begin{eqnarray*}
\lambda^*&:=& \underset{K\subseteq I}{\bigwedge} \Bigg(\mathtt{max}_K\rightarrow \Box \bigg(\mathtt{info}_K \rightarrow \underset{S\subseteq J}{\bigwedge} \Big(\mathtt{state}_S\rightarrow \underset{s\in S}{\bigvee}\big(\Box (\mathtt{info}_K\rightarrow \beta_s)\wedge \underset{m \in D_s}{\bigwedge}\Diamond(\mathtt{info}_K \wedge \gamma_m) \big) \Big)\bigg) \Bigg).\end{eqnarray*}
As we show in Lemma \ref{Lem1}.\ref{Lem1c} below, if $\mathtt{max}_K$ is true, then (the truth set of) $\mathtt{info}_K$ picks out a maximal $\Theta$-accepting subset of the current information state. The translation thus guarantees that if we restrict the current information state to such a subset, then at least one of the disjuncts of $\Omega$ will be true at each world in the restricted information state. In particular, for at least one such formula $\omega_s$ we have that $\beta_s$ is true throughout the subset, so $\Box \beta_s$ is true, and each of the $\Diamond \gamma_m$ conjuncts of $\omega_s$ is witnessed by some world making $\gamma_m$ true. We also need to know that these formulas $\mathtt{info}_K$ for which $\mathtt{max}_K$ holds pick out \emph{all} of the maximal $\Theta$-accepting subsets, which is the content of Lemma \ref{Lem1}.\ref{Lem1b}.

We regard $(\cdot)^*$ as a partial function from $\mathcal{L}(\Rightarrow)$ to $\mathcal{L}$ such that $\varphi^*$ is defined iff $\varphi$ is of the form $\Theta\Rightarrow \Omega$ as above. We then define a partial function $(\cdot)^\dagger$ from $\mathcal{L}(\Rightarrow)$ to $\mathcal{L}$ as follows: 
\begin{itemize}
\item $p^\dagger = p$; $(\neg\varphi)^\dagger = \neg \varphi^\dagger$; $(\varphi\wedge\psi)^\dagger = \varphi^\dagger\wedge\psi^\dagger$; $(\Box\varphi)^\dagger = \Box \varphi^\dagger$;
\item $(\varphi \Rightarrow \psi)^\dagger =\begin{cases} \big((\varphi^\dagger)^{NF}\Rightarrow(\psi^\dagger)^{NF}\big)^* & \mbox{if this is defined} \\ \mbox{undefined} & \mbox{otherwise} \end{cases}$,
\end{itemize} 
where $(\chi)^{NF}$ is the \textsf{K45} disjunctive normal form of $\chi$. An easy induction shows that $(\cdot)^\dagger$ is in fact a total function, so 
\begin{itemize}
\item $(\varphi \Rightarrow \psi)^\dagger = \big((\varphi^\dagger)^{NF}\Rightarrow(\psi^\dagger)^{NF}\big)^*$,
\end{itemize}
and $\varphi^\dagger\in\mathcal{L}$ for every $\varphi\in\mathcal{L}(\Rightarrow)$.

\begin{theorem}[Reduction of $\mathcal{L}(\Rightarrow)$ to $\mathcal{L}$]\label{ReductionThm} For every $\delta\in\mathcal{L}(\Rightarrow)$:
\begin{enumerate}
\item $\delta^\dagger\in\mathcal{L}$;
\item for every pointed model $\mathcal{M},w,X$: $\mathcal{M},w,X\vDash \delta\mbox{ iff }\mathcal{M},w,X\vDash \delta^\dagger$.
\end{enumerate}
\end{theorem}

\begin{proof} The proof is by induction on $\delta$. The only nontrivial case is where $\delta$ is of the form $\varphi\Rightarrow \psi$, so we must show $\mathcal{M},w,X\vDash \varphi\Rightarrow\psi$ iff $\mathcal{M},w,X\vDash \big((\varphi^\dagger)^{NF}\Rightarrow(\psi^\dagger)^{NF}\big)^*$. By the inductive hypothesis, $\llbracket \varphi\rrbracket^{\mathcal{M},X}=\llbracket \varphi^\dagger\rrbracket^{\mathcal{M},X}$, and by Lemma \ref{NormalForm}, since $\varphi^\dagger\in\mathcal{L}$, $\llbracket \varphi^\dagger\rrbracket^{\mathcal{M},X}=\llbracket (\varphi^\dagger)^{NF}\rrbracket^{\mathcal{M},X}$. Thus, $\llbracket \varphi\rrbracket^{\mathcal{M},X}=\llbracket (\varphi^\dagger)^{NF}\rrbracket^{\mathcal{M},X}$, and similarly $\llbracket \psi\rrbracket^{\mathcal{M},X}=\llbracket (\psi^\dagger)^{NF}\rrbracket^{\mathcal{M},X}$. It follows that $\mathcal{M},w,X\vDash \varphi\Rightarrow\psi$ iff $\mathcal{M},w,X\vDash (\varphi^\dagger)^{NF}\Rightarrow(\psi^\dagger)^{NF}$. To complete the proof, it only remains to show that for any $\Theta\Rightarrow\Omega$ where $\Theta,\Omega\in\mathcal{L}$ are in normal form, we have $\mathcal{M},w,X\vDash \Theta\Rightarrow\Omega$ iff $\mathcal{M},w,X\vDash (\Theta\Rightarrow\Omega)^*$, as in Proposition \ref{mainProp} below. 
\end{proof}

As the proof of Theorem \ref{ReductionThm} shows, the key task is now to prove Proposition \ref{mainProp}, for which we need a preliminary lemma. From now on, we say that $Y\subseteq X$ is a \textit{$\Theta$-subset} of $X$ if $Y\subseteq \llbracket \Theta\rrbracket^{\mathcal{M},Y}$. 

\begin{lemma}\label{Lem1} For any pointed model $\mathcal{M},w,X$:
\begin{enumerate}
\item\label{Lem1a} if $\mathcal{M},w,X\vDash \mathtt{good}_L$, then $\llbracket \mathtt{info}_L\rrbracket^{\mathcal{M},X}$ is a $\Theta$-subset of $X$;
\item\label{Lem1b} if $Y$ is a maximal $\Theta$-subset of $X$, there is a $K\subseteq I$ such that $Y=\llbracket \mathtt{info}_K\rrbracket^{\mathcal{M},X}$ and $\mathcal{M},w,X\vDash \mathtt{max}_K$;
\item\label{Lem1c} if $\mathcal{M},w,X\vDash \mathtt{max}_L$, then $\llbracket \mathtt{info}_L\rrbracket^{\mathcal{M},X}$ is a maximal $\Theta$-subset of $X$.
\end{enumerate}
\end{lemma}
\begin{proof} For part \ref{Lem1a}, suppose $\mathcal{M},w,X\vDash \mathtt{good}_L$. We must show that $\llbracket \mathtt{info}_L\rrbracket^{\mathcal{M},X}\subseteq\llbracket \Theta\rrbracket^{\mathcal{M}, \llbracket \mathtt{info}_L\rrbracket^{\mathcal{M},X}}$. So suppose $v\in \llbracket \mathtt{info}_L\rrbracket^{\mathcal{M},X}$. Then there is some $p\in L$ such that $v\in \llbracket \varphi_p\rrbracket^{\mathcal{M},X}$. We claim that $v\in\llbracket \theta_p \rrbracket ^{\mathcal{M}, \llbracket \mathtt{info}_L\rrbracket^{\mathcal{M},X}}$. By definition of $\mathtt{info}_L$, we have $\llbracket\mathtt{info}_L\rrbracket^{\mathcal{M},X}\subseteq \llbracket \underset{l\in L}{\bigwedge}\psi_l\rrbracket ^{\mathcal{M},X}$ and hence $\llbracket\mathtt{info}_L\rrbracket^{\mathcal{M},X}\subseteq \llbracket \underset{l\in L}{\bigwedge}\psi_l\rrbracket^{\mathcal{M}, \llbracket \mathtt{info}_L\rrbracket^{\mathcal{M},X}}$ since the $\psi_l$'s are nonmodal, so 
\begin{equation}\llbracket\mathtt{info}_L\rrbracket^{\mathcal{M},X}\subseteq \llbracket \underset{l\in L}{\bigwedge}\Box\psi_l\rrbracket^{\mathcal{M}, \llbracket \mathtt{info}_L\rrbracket^{\mathcal{M},X}}.\label{Eq1}\end{equation}
 Then since $\mathcal{M},w,X\vDash \mathtt{good}_L$, we also have \[\mathcal{M},w,X\vDash \underset{n\in D_p}{\bigwedge}\Diamond(\mathtt{info}_L\wedge \chi_n),\] which with the fact that the $\chi_n$'s are nonmodal implies  
\begin{equation}\llbracket\mathtt{info}_L\rrbracket^{\mathcal{M},X}\subseteq \llbracket \underset{n\in D_p}{\bigwedge}\Diamond\chi_n\rrbracket^{\mathcal{M}, \llbracket \mathtt{info}_L\rrbracket^{\mathcal{M},X}}.\label{Eq2}\end{equation} 
By $v\in \llbracket \varphi_p\rrbracket^{\mathcal{M},X}$, (\ref{Eq1}), and (\ref{Eq2}), we have $v\in\llbracket \theta_p \rrbracket ^{\mathcal{M}, \llbracket \mathtt{info}_L\rrbracket^{\mathcal{M},X}}$. Hence $\llbracket \mathtt{info}_L\rrbracket^{\mathcal{M},X}\subseteq\llbracket \Theta\rrbracket^{\mathcal{M}, \llbracket \mathtt{info}_L\rrbracket^{\mathcal{M},X}}$.
 
 For part \ref{Lem1b}, suppose $Y$ is a maximal $\Theta$-subset of $X$, so $Y\subseteq\llbracket \Theta\rrbracket^{\mathcal{M},Y}$ and there is no $Z$ such that $Y\subsetneq Z\subseteq X$ and $Z\subseteq\llbracket \Theta\rrbracket^{\mathcal{M},Z}$. Let 
\begin{equation}K= \{k\in I\mid  \llbracket \theta_k\rrbracket^{\mathcal{M},Y}\not=\varnothing\}.\label{Eq3}\end{equation}
We will show that $Y\subseteq\llbracket \mathtt{info}_K\rrbracket^{\mathcal{M},X}$. First, observe that for each $k\in K$, since $\llbracket \theta_k\rrbracket^{\mathcal{M},Y}\not=\varnothing$, we have $Y\subseteq\llbracket \psi_k\rrbracket^{\mathcal{M},X}$, so $Y\subseteq\llbracket \underset{k\in K}{\bigwedge}\psi_k\rrbracket^{\mathcal{M},X}$. Moreover, since $Y\subseteq\llbracket \Theta\rrbracket^{\mathcal{M},Y}$, for every $y\in Y$, there is a $k_y\in K$ with $y\in \llbracket \theta_{k_y}\rrbracket^{\mathcal{M},Y}$, so $y\in \llbracket \varphi_{k_y}\rrbracket^{\mathcal{M},X}$. Thus, for every $y\in Y$, $y\in \llbracket \varphi_{k_y}\wedge \underset{k\in K}{\bigwedge}\psi_k\rrbracket^{\mathcal{M},X}$. Hence $Y\subseteq\llbracket \mathtt{info}_K\rrbracket^{\mathcal{M},X}$.\\

Next we show that $Y\supseteq\llbracket \mathtt{info}_K\rrbracket^{\mathcal{M},X}$. Suppose not, so there is a $w\in \llbracket \mathtt{info}_K\rrbracket^{\mathcal{M},X}$ such that $w\not\in Y$. Since $w\in \llbracket \mathtt{info}_K\rrbracket^{\mathcal{M},X}$, we have  $w\in\llbracket \underset{k\in K}{\bigwedge} \psi_k \rrbracket^{\mathcal{M},X}$. Now we claim that $Y\cup\{w\}\subseteq \llbracket \Theta\rrbracket^{\mathcal{M},Y\cup \{w\}}$. Consider a $y\in Y$ and a disjunct $\theta_r$ of $\Theta$ such that $y\in\llbracket \theta_r\rrbracket^{\mathcal{M},Y}$, which exists since $Y\subseteq\llbracket \Theta\rrbracket^{\mathcal{M},Y}$. We claim that $y\in\llbracket \theta_r\rrbracket^{\mathcal{M},Y\cup\{w\}}$. First, clearly $y\in \llbracket \varphi_r \wedge \underset{n\in D_r}{\bigwedge}\Diamond\chi_n\rrbracket^{\mathcal{M},Y}$ implies $y\in\llbracket \varphi_r \wedge \underset{n\in D_r}{\bigwedge}\Diamond\chi_n\rrbracket^{\mathcal{M},Y\cup\{w\}}$. Then since $w\in\llbracket \underset{k\in K}{\bigwedge} \psi_k \rrbracket^{\mathcal{M},X}$, $y\in\llbracket \Box\psi_r\rrbracket^{\mathcal{M},Y}$ implies $y\in\llbracket \Box\psi_r\rrbracket^{\mathcal{M},Y\cup\{w\}}$. Thus, $y\in\llbracket \theta_r\rrbracket^{\mathcal{M},Y\cup\{w\}}$ and hence $y\in\llbracket \Theta\rrbracket^{\mathcal{M},Y\cup\{w\}}$, so we have shown that $Y\subseteq \llbracket \Theta\rrbracket^{\mathcal{M},Y\cup \{w\}}$. It only remains to observe that $w\in \llbracket \Theta\rrbracket^{\mathcal{M},Y\cup \{w\}}$. Since $w\in \llbracket \mathtt{info}_K\rrbracket^{\mathcal{M},X}$, there is a $t\in K$ such that $w\in\llbracket \varphi_t\rrbracket^{\mathcal{M},X}$. It is then easy to see that $w\in \llbracket \theta_t\rrbracket^{\mathcal{M},Y\cup\{w\}}$. Hence $Y\cup\{w\}\subseteq \llbracket \Theta\rrbracket^{\mathcal{M},Y\cup \{w\}}$, contradicting the assumption that $Y$ is a \textit{maximal} $\Theta$-subset of $X$.

Thus, we have shown that $Y=\llbracket \mathtt{info}_K\rrbracket^{\mathcal{M},X}$. Finally, we must show that $\mathcal{M},w,X\vDash \mathtt{max}_K$. To show that $\mathcal{M},w,X\vDash \mathtt{good}_K$, it follows from (\ref{Eq3}) that for every $k\in K$ and $n\in D_k$, $Y\cap \llbracket \gamma_n\rrbracket^{\mathcal{M},X}\not=\varnothing$, which with $Y=\llbracket \mathtt{info}_K\rrbracket^{\mathcal{M},X}$ gives us  $\llbracket \mathtt{info}_K\rrbracket^{\mathcal{M},X}\cap \llbracket \gamma_n\rrbracket^{\mathcal{M},X}\not=\varnothing$, so $\mathcal{M},w,X\vDash \mathtt{good}_K$. Finally, suppose $\mathcal{M},w,X\vDash \Box(\mathtt{info}_K\rightarrow \mathtt{info}_L)\wedge \Diamond (\neg \mathtt{info}_K\wedge\mathtt{info}_L)$, so $\llbracket \mathtt{info}_L\rrbracket^{\mathcal{M},X}\supsetneq \llbracket \mathtt{info}_K\rrbracket^{\mathcal{M},X}$. Then since $Y=\llbracket \mathtt{info}_K\rrbracket^{\mathcal{M},X}$ is a \textit{maximal} $\Theta$-subset of $X$, $\llbracket \mathtt{info}_L\rrbracket^{\mathcal{M},X}$ is not a $\Theta$-subset of $X$. Hence by part \ref{Lem1a}, we have $\mathcal{M},w,X\vDash \neg \mathtt{good}_L$. This completes the proof that $\mathcal{M},w,X\vDash \mathtt{max}_K$.

For part \ref{Lem1c}, if $\mathcal{M},w,X\vDash \mathtt{max}_L$, then $\mathcal{M},w,X\vDash \mathtt{good}_L$, so by part \ref{Lem1a},  $\llbracket \mathtt{info}_L\rrbracket^{\mathcal{M},X}$ is a $\Theta$-subset of $X$. Now suppose $\llbracket \mathtt{info}_L\rrbracket^{\mathcal{M},X}$ is not a \textit{maximal} $\Theta$-subset, so there is a $\Theta$-subset $Y\supsetneq \llbracket \mathtt{info}_L\rrbracket^{\mathcal{M},X}$. Then by part \ref{Lem1b}, there is a $K\subseteq I$ such that $Y=\llbracket \mathtt{info}_K\rrbracket^{\mathcal{M},X}$ and $\mathcal{M},w,X\vDash \mathtt{max}_K$, so $\mathcal{M},w,X\vDash \mathtt{good}_K$. Since $Y\supsetneq \llbracket \mathtt{info}_L\rrbracket^{\mathcal{M},X}$ and $Y=\llbracket \mathtt{info}_K\rrbracket^{\mathcal{M},X}$, we have $\llbracket\mathtt{info}_K\rrbracket^{\mathcal{M},X}\supsetneq \llbracket\mathtt{info}_L\rrbracket^{\mathcal{M},X}$, so $\mathcal{M},w,X\vDash \Box(\mathtt{info}_L\rightarrow \mathtt{info}_K)\wedge \Diamond (\neg \mathtt{info}_L\wedge\mathtt{info}_K))$. Then since $\mathcal{M},w,X\vDash \mathtt{max}_L$, it follows that $\mathcal{M},w,X\vDash \neg \mathtt{good}_K$. From this contradiction we conclude that $\llbracket \mathtt{info}_L\rrbracket^{\mathcal{M},X}$ is a maximal $\Theta$-subset of $X$.\end{proof}

We are now ready to establish the key proposition used in the proof of Theorem \ref{ReductionThm}, namely the semantic equivalence of a conditional $\lambda$ (without embedded conditionals) and its translation $\lambda^*$.

\begin{proposition}\label{mainProp}  For any pointed model $\mathcal{M},w,X$: \[\mathcal{M},w,X\vDash \lambda\mbox{ iff }\mathcal{M},w,X\vDash \lambda^*.\]
\end{proposition}

\begin{proof} Since $\lambda:= \Theta\Rightarrow\Omega$, we have $\mathcal{M},w,X\vDash \lambda$ iff for all maximal $\Theta$-subsets $Y$ of $X$, $Y\subseteq \llbracket \Omega\rrbracket^{\mathcal{M},Y}$.
By Lemma \ref{Lem1}, $Y$ is a maximal $\Theta$-subset of $X$ iff there is a $K\subseteq I$ such that $Y=\llbracket \mathtt{info}_K\rrbracket^{\mathcal{M},X}$ and $\mathcal{M},w,X\vDash \mathtt{max}_K$. Thus, the condition that $\mathcal{M},w,X\vDash \lambda^*$, i.e.,
\begin{eqnarray*}
\mathcal{M},w,X\vDash\underset{K\subseteq I}{\bigwedge} \Bigg(\mathtt{max}_K\rightarrow \Box \bigg(\mathtt{info}_K \rightarrow \underset{S\subseteq J}{\bigwedge} \Big(\mathtt{state}_S\rightarrow \underset{s\in S}{\bigvee}\big(\Box (\mathtt{info}_K\rightarrow \beta_s)\wedge \underset{m \in D_s}{\bigwedge}\Diamond(\mathtt{info}_K \wedge \gamma_m) \big) \Big)\bigg) \Bigg)\label{Eq3.5}\end{eqnarray*}
 is equivalent to: 
\begin{eqnarray}
&& \mbox{for all maximal $\Theta$-subsets $Y$ of $X$, there is a $K\subseteq I$ such that  $Y=\llbracket \mathtt{info}_K\rrbracket^{\mathcal{M},X}$ and}  \nonumber\\
&& \mathcal{M},w,X\vDash \Box \bigg(\mathtt{info}_K \rightarrow \underset{S\subseteq J}{\bigwedge} \Big(\mathtt{state}_S\rightarrow \underset{s\in S}{\bigvee}\big(\Box (\mathtt{info}_K\rightarrow \beta_s)\wedge \underset{m \in D_s}{\bigwedge}\Diamond(\mathtt{info}_K \wedge \gamma_m) \big) \Big)\bigg). \label{Eq4}\end{eqnarray}
Below we will show that if  $Y=\llbracket \mathtt{info}_K\rrbracket^{\mathcal{M},X}$, then 
\begin{eqnarray}
&&\mathcal{M},w,X\vDash \Box \bigg(\mathtt{info}_K \rightarrow \underset{S\subseteq J}{\bigwedge} \Big(\mathtt{state}_S\rightarrow \underset{s\in S}{\bigvee}\big(\Box (\mathtt{info}_K\rightarrow \beta_s)\wedge \underset{m \in D_s}{\bigwedge}\Diamond(\mathtt{info}_K \wedge \gamma_m) \big) \Big)\bigg)  \label{Eq5}
\end{eqnarray}
iff $Y\subseteq \llbracket \Omega\rrbracket^{\mathcal{M},Y}$. Thus, (\ref{Eq4}) is equivalent to: for all maximal $\Theta$-subsets $Y$ of $X$, $Y\subseteq \llbracket \Omega\rrbracket^{\mathcal{M},Y}$. Given the other equivalences above, this establishes $\mathcal{M},w,X\vDash \lambda^*$ iff $\mathcal{M},w,X\vDash \lambda$.

Now suppose $Y=\llbracket \mathtt{info}_K\rrbracket^{\mathcal{M},X}$. Then (\ref{Eq5}) is equivalent 
\begin{eqnarray*}
Y&\subseteq& \Bigg\llbracket \underset{S\subseteq J}{\bigwedge} \Big(\mathtt{state}_S\rightarrow \underset{s\in S}{\bigvee}\big(\Box (\mathtt{info}_K\rightarrow \beta_s)\wedge \underset{m \in D_s}{\bigwedge}\Diamond(\mathtt{info}_K \wedge \gamma_m) \big) \Big) \Bigg\rrbracket^{\mathcal{M},X}.
\end{eqnarray*}
Thus, to show $Y\subseteq \llbracket \Omega\rrbracket^{\mathcal{M},Y}$, it suffices to show
\begin{eqnarray}
&&\Bigg\llbracket \underset{S\subseteq J}{\bigwedge} \Big(\mathtt{state}_S\rightarrow \underset{s\in S}{\bigvee}\big(\Box (\mathtt{info}_K\rightarrow \beta_s)\wedge \underset{m \in D_s}{\bigwedge}\Diamond(\mathtt{info}_K \wedge \gamma_m) \big) \Big) \Bigg\rrbracket^{\mathcal{M},X} \subseteq \llbracket \Omega\rrbracket^{\mathcal{M},Y}.\label{Eq6}
\end{eqnarray}
So suppose $v$ is in the left-hand side of (\ref{Eq6}). By definition of $\mathtt{state}_S$, there is exactly one $S\subseteq J$ such that $v\in \llbracket \mathtt{state}_S\rrbracket^{\mathcal{M},X}$. Moreover, $S\not=\varnothing$, for if $S=\varnothing$, then the empty disjunction in (\ref{Eq6}) is $\bot$, so we would have $v\in \llbracket \mathtt{state}_S\rrbracket^{\mathcal{M},X}$ and $v\in \llbracket \mathtt{state}_S\rightarrow \bot\rrbracket^{\mathcal{M},X}=\llbracket \neg \mathtt{state}_S\rrbracket^{\mathcal{M},X}$, i.e.,  $v\not\in \llbracket \mathtt{state}_S\rrbracket^{\mathcal{M},X}$, a contradiction. Since $S\not=\varnothing$, there is an $s\in S$ such that 
\begin{equation}
v\in \llbracket \alpha_s\wedge\Box(\mathtt{info}_K\rightarrow \beta_s)\wedge \underset{m \in D_s}{\bigwedge}\Diamond(\mathtt{info}_K \wedge \gamma_m)  \rrbracket^{\mathcal{M},X}
\end{equation}
which with $Y=\llbracket \mathtt{info}_K\rrbracket^{\mathcal{M},X}$ implies
\begin{equation}
v\in \llbracket \alpha_s\wedge \Box \beta_s\wedge \underset{m \in D_s}{\bigwedge}\Diamond \gamma_m \rrbracket^{\mathcal{M},Y},
\end{equation}
so $v\in\llbracket \omega_s\rrbracket^{\mathcal{M},Y}\subseteq \llbracket \Omega\rrbracket^{\mathcal{M},Y}$. This establishes (\ref{Eq6}), which completes the proof. \end{proof}

Although we have now shown that any formula of $\mathcal{L}(\Rightarrow)$ can be effectively translated into a modal formula that is semantically equivalent according to Kolodny and MacFarlane's semantics, there is clearly a huge blowup in formula size. It is reasonable to conjecture that according to this semantics, $\mathcal{L}(\Rightarrow)$ is \textit{exponentially more succinct} than the basic modal language $\mathcal{L}$ in the sense of \cite{Lutz2006}. 

\section{Conclusion}

In \S\ref{DELtoLanguage}, we presented an example of how techniques from dynamic epistemic logic can be fruitfully applied to the formal semantics of modals and conditionals, by providing a complete axiomatization of the inferences validated by a formal semantics. For arguments that knowing such a complete axiomatization is of value for formal semantics, see \cite{HIForthcoming}. In the other direction, in \S\ref{SemanticsToDEL}, we presented an example of how ideas from the formal semantics of modals and conditionals can be profitably imported into dynamic epistemic logic, by enabling a natural kind of reasoning about epistemic change---focusing not on the mechanism of epistemic change but rather on the desired result of epistemic change. We hope that these examples might provide some stimulus for further cross-pollination between these two overlapping fields.

\nocite{*}
\bibliographystyle{eptcs}
\bibliography{indicativesDEL}

\providecommand{\noopsort}[1]{}
\begin{thebibliography}{10}
\providecommand{\bibitemdeclare}[2]{}
\providecommand{\surnamestart}{}
\providecommand{\surnameend}{}
\providecommand{\urlprefix}{Available at }
\providecommand{\url}[1]{\texttt{#1}}
\providecommand{\href}[2]{\texttt{#2}}
\providecommand{\urlalt}[2]{\href{#1}{#2}}
\providecommand{\doi}[1]{doi:\urlalt{http://dx.doi.org/#1}{#1}}
\providecommand{\bibinfo}[2]{#2}

\bibitemdeclare{article}{Benthem2008}
\bibitem{Benthem2008}
\bibinfo{author}{Johan \surnamestart {\noopsort{Benthem}}{van
  Benthem}\surnameend} (\bibinfo{year}{2008}): \emph{\bibinfo{title}{{Logical
  Dynamics Meets Logical Pluralism?}}}
\newblock {\sl \bibinfo{journal}{Australasian Journal of Logic}}
  \bibinfo{volume}{6}, pp. \bibinfo{pages}{182--209}.

\bibitemdeclare{book}{Benthem2011}
\bibitem{Benthem2011}
\bibinfo{author}{Johan \surnamestart {\noopsort{Benthem}}{van
  Benthem}\surnameend} (\bibinfo{year}{2011}): \emph{\bibinfo{title}{Logical
  Dynamics of Information and Interaction}}.
\newblock \bibinfo{publisher}{Cambridge University Press},
  \doi{10.1017/CBO9780511974533}.

\bibitemdeclare{book}{Blackburn2001}
\bibitem{Blackburn2001}
\bibinfo{author}{Patrick \surnamestart Blackburn\surnameend},
  \bibinfo{author}{Maarten \surnamestart de~Rijke\surnameend} \&
  \bibinfo{author}{Yde \surnamestart Venema\surnameend} (\bibinfo{year}{2001}):
  \emph{\bibinfo{title}{Modal Logic}}.
\newblock \bibinfo{publisher}{Cambridge University Press},
  \doi{10.1017/CBO9781107050884}.

\bibitemdeclare{article}{Bledin2014}
\bibitem{Bledin2014}
\bibinfo{author}{Justin \surnamestart Bledin\surnameend}
  (\bibinfo{year}{2014}): \emph{\bibinfo{title}{{Logic Informed}}}.
\newblock {\sl \bibinfo{journal}{Mind}}
  \bibinfo{volume}{123}(\bibinfo{number}{490}), pp. \bibinfo{pages}{277--316},
  \doi{10.1093/mind/fzu073}.

\bibitemdeclare{book}{Ditmarsch2008}
\bibitem{Ditmarsch2008}
\bibinfo{author}{Hans \surnamestart {\noopsort{Ditmarsch}}{van
  Ditmarsch}\surnameend}, \bibinfo{author}{Wiebe \surnamestart van~der
  Hoek\surnameend} \& \bibinfo{author}{Barteld \surnamestart Kooi\surnameend}
  (\bibinfo{year}{2008}): \emph{\bibinfo{title}{{Dynamic Epistemic Logic}}}.
\newblock \bibinfo{publisher}{Springer}, \doi{10.1007/978-1-4020-5839-4}.

\bibitemdeclare{incollection}{HI2010}
\bibitem{HI2010}
\bibinfo{author}{Wesley~H. \surnamestart Holliday\surnameend} \&
  \bibinfo{author}{Thomas~F. \surnamestart Icard\surnameend, III}
  (\bibinfo{year}{2010}): \emph{\bibinfo{title}{{Moorean Phenomena in Epistemic
  Logic}}}.
\newblock In \bibinfo{editor}{L.~\surnamestart Beklemishev\surnameend},
  \bibinfo{editor}{V.~\surnamestart Goranko\surnameend} \&
  \bibinfo{editor}{V.~\surnamestart Shehtman\surnameend}, editors: {\sl
  \bibinfo{booktitle}{Advances in Modal Logic}}, \bibinfo{volume}{8},
  \bibinfo{publisher}{College Publications}, pp. \bibinfo{pages}{178--199}.

\bibitemdeclare{inproceedings}{HI2013}
\bibitem{HI2013}
\bibinfo{author}{Wesley~H. \surnamestart Holliday\surnameend} \&
  \bibinfo{author}{Thomas~F. \surnamestart Icard\surnameend, III}
  (\bibinfo{year}{2013}): \emph{\bibinfo{title}{Measure semantics and
  qualitative semantics for epistemic modals}}.
\newblock In \bibinfo{editor}{T.~\surnamestart Snider\surnameend}, editor: {\sl
  \bibinfo{booktitle}{Proceedings of SALT 23}}, \bibinfo{publisher}{LSA and CLC
  Publications}, pp. \bibinfo{pages}{514--534}, \doi{10.3765/salt.v23i0.2670}.

\bibitemdeclare{incollection}{HIForthcoming}
\bibitem{HIForthcoming}
\bibinfo{author}{Wesley~H. \surnamestart Holliday\surnameend} \&
  \bibinfo{author}{Thomas~F. \surnamestart Icard\surnameend, III}
  (\bibinfo{year}{Forthcoming}): \emph{\bibinfo{title}{{Axiomatization in the
  Meaning Sciences}}}.
\newblock In \bibinfo{editor}{D.~\surnamestart Ball\surnameend} \&
  \bibinfo{editor}{B.~\surnamestart Rabern\surnameend}, editors: {\sl
  \bibinfo{booktitle}{The Science of Meaning}}, \bibinfo{publisher}{Oxford
  University Press}.

\bibitemdeclare{article}{MacFarlane2010}
\bibitem{MacFarlane2010}
\bibinfo{author}{Niko \surnamestart Kolodny\surnameend} \&
  \bibinfo{author}{John \surnamestart MacFarlane\surnameend}
  (\bibinfo{year}{2010}): \emph{\bibinfo{title}{{Ifs and Oughts}}}.
\newblock {\sl \bibinfo{journal}{Journal of Philosophy}} \bibinfo{volume}{107},
  pp. \bibinfo{pages}{115--143}, \doi{10.5840/jphil2010107310}.

\bibitemdeclare{inproceedings}{Lutz2006}
\bibitem{Lutz2006}
\bibinfo{author}{Carsten \surnamestart Lutz\surnameend} (\bibinfo{year}{2006}):
  \emph{\bibinfo{title}{Complexity and succinctness of public announcement
  logic}}.
\newblock In \bibinfo{editor}{P.~\surnamestart Stone\surnameend} \&
  \bibinfo{editor}{G.~\surnamestart Weiss\surnameend}, editors: {\sl
  \bibinfo{booktitle}{Proceedings of the Fifth International Joint Conference
  on Autonomous Agents and Multiagent Systems (AAMAS 06)}},
  \bibinfo{publisher}{ACM}, pp. \bibinfo{pages}{137--143},
  \doi{10.1145/1160633.1160657}.

\bibitemdeclare{incollection}{McLaughlin2014}
\bibitem{McLaughlin2014}
\bibinfo{author}{Brian \surnamestart McLaughlin\surnameend} \&
  \bibinfo{author}{Karen \surnamestart Bennett\surnameend}
  (\bibinfo{year}{2014}): \emph{\bibinfo{title}{Supervenience}}.
\newblock In \bibinfo{editor}{E.~N. \surnamestart Zalta\surnameend}, editor:
  {\sl \bibinfo{booktitle}{The Stanford Encyclopedia of Philosophy}},
  \bibinfo{edition}{spring 2014} edition, \bibinfo{publisher}{Metaphysics
  Research Lab, Stanford University}.

\bibitemdeclare{article}{Schulz2010}
\bibitem{Schulz2010}
\bibinfo{author}{Moritz \surnamestart Schulz\surnameend}
  (\bibinfo{year}{2010}): \emph{\bibinfo{title}{Epistemic modals and
  informational consequence}}.
\newblock {\sl \bibinfo{journal}{Synthese}}
  \bibinfo{volume}{174}(\bibinfo{number}{3}), pp. \bibinfo{pages}{385--395},
  \doi{10.1007/s11229-009-9461-8}.

\bibitemdeclare{incollection}{Vaananen2008}
\bibitem{Vaananen2008}
\bibinfo{author}{Jouko \surnamestart V\"{a}\"{a}n\"{a}nen\surnameend}
  (\bibinfo{year}{2008}): \emph{\bibinfo{title}{{Modal Dependence Logic}}}.
\newblock In \bibinfo{editor}{K.~R. \surnamestart Apt\surnameend} \&
  \bibinfo{editor}{R.~\surnamestart van Rooij\surnameend}, editors: {\sl
  \bibinfo{booktitle}{New Perspectives on Games and Interaction}}, {\sl
  \bibinfo{series}{Texts in Logic and Games}}~\bibinfo{volume}{4},
  \bibinfo{publisher}{Amsterdam University Press}, pp.
  \bibinfo{pages}{237--254}, \doi{10.5117/9789089640574}.

\bibitemdeclare{article}{Willer2012}
\bibitem{Willer2012}
\bibinfo{author}{Malte \surnamestart Willer\surnameend} (\bibinfo{year}{2012}):
  \emph{\bibinfo{title}{{A Remark on Iffy Oughts}}}.
\newblock {\sl \bibinfo{journal}{The Journal of Philosophy}}
  \bibinfo{volume}{109}(\bibinfo{number}{7}), pp. \bibinfo{pages}{449--461},
  \doi{10.5840/jphil2012109719}.

\bibitemdeclare{article}{Yalcin2007}
\bibitem{Yalcin2007}
\bibinfo{author}{Seth \surnamestart Yalcin\surnameend} (\bibinfo{year}{2007}):
  \emph{\bibinfo{title}{{Epistemic Modals}}}.
\newblock {\sl \bibinfo{journal}{Mind}}
  \bibinfo{volume}{116}(\bibinfo{number}{464}), pp. \bibinfo{pages}{983--1026},
  \doi{10.1093/mind/fzm983}.

\bibitemdeclare{article}{Yalcin2012a}
\bibitem{Yalcin2012a}
\bibinfo{author}{Seth \surnamestart Yalcin\surnameend} (\bibinfo{year}{2012}):
  \emph{\bibinfo{title}{{A Counterexample to Modus Tollens}}}.
\newblock {\sl \bibinfo{journal}{Journal of Philosophical Logic}}
  \bibinfo{volume}{41}(\bibinfo{number}{6}), pp. \bibinfo{pages}{1001--1024},
  \doi{10.1007/s10992-012-9228-4}.

\end{thebibliography}
\end{document}